\documentclass[]{interact}

\usepackage{epstopdf}
\usepackage[caption=false]{subfig}

\usepackage[numbers,sort&compress]{natbib}
\bibpunct[, ]{[}{]}{,}{n}{,}{,}
\renewcommand\bibfont{\fontsize{10}{12}\selectfont}
\makeatletter
\def\NAT@def@citea{\def\@citea{\NAT@separator}}
\makeatother

\theoremstyle{plain}
\newtheorem{theorem}{Theorem}
\newtheorem{lemma}[theorem]{Lemma}

\theoremstyle{definition}

\theoremstyle{remark}

\usepackage{amsmath,amsthm} 
\usepackage{booktabs}
\usepackage{verbatim}
\usepackage{geometry}
\usepackage{bm}
\usepackage{wasysym}
\usepackage{esint}
\usepackage{xcolor}
\usepackage{tikz}
\usetikzlibrary{shapes,backgrounds,arrows,chains,matrix,positioning,scopes}
\usepackage{caption}
\usepackage{titlesec}

\usepackage{multirow}
\usepackage{natbib}
\usepackage{rotating}
\usepackage{lscape}
\begin{document}
	
	
	\title{Type I-Type II Mixture Censoring Scheme for Lifetime Data Analysis}
	
	\author{
		\name{K.~K. Anakha\textsuperscript{a}\thanks{CONTACT K.~K. Author. Email: anakhakk@stthomas.ac.in} and V.~M. Chacko\textsuperscript{a}}
		\affil{\textsuperscript{a}Department of Statistics, St.Thomas College (Autonomous), Thrissur, University of Calicut, Kerala, India}
	}
	
	\maketitle
	\begin{abstract}
		The Type-I and Type-II censoring schemes are the most prominent and commonly used censoring schemes in practice. In this work, a mixture of Type-I and Type-II censoring schemes, named the Type I-Type II mixture censoring scheme, has been introduced. Different censoring schemes have been discussed, along with their benefits and drawbacks. For the proposed censoring scheme, we analyze the data under the assumption that failure times of experimental units follow the Weibull distribution. The computational formulas for the expected number of failures and the expected failure time are provided. Maximum likelihood estimation and Bayesian estimation are used to estimate the model parameters. On the basis of thorough Monte Carlo simulations, the investigated inferential approaches are evaluated. Finally, a numerical example is provided to demonstrate the method of inference discussed here. 
	\end{abstract}
	

\section{Introduction}
Effective censoring schemes are crucial in the field of statistical analysis, as they play a vital role in accurately evaluating and deriving significant insights from censored data. Censoring, a common phenomenon in many fields such as survival analysis, reliability engineering, and clinical trials, occurs when the complete information about the event of interest is not available for some observations. Traditional censoring schemes (Type-I and Type-II), while widely used, may not capture the full potential of the available data or adequately handle complex scenarios. 

In this article, we delve into a new censoring scheme that promises to advance statistical inference techniques and unlock new avenues of knowledge. This innovative approach challenges conventional methodologies, offering researchers a fresh perspective on handling censored data. By introducing a novel framework for censoring, we aim to bridge the gap between existing methods and the increasing complexity of real-world datasets, enabling more accurate and insightful analysis.

The conventional censoring schemes, Type-I and Type-II, have specific features. Type-I censoring scheme involves running the experiment until a fixed time point $T$ is reached, while Type-II censoring scheme allows the experiment to continue until the occurrence of the $m$th failure. These schemes lack complete event time information for censored observations. This loss of information can limit the precision and accuracy of statistical analysis, potentially leading to biased estimates or reduced statistical power.

To address the drawback, Epstein\cite{epstein1954truncated} introduced, a mixture of Type-I and Type-II censoring scheme, known as hybrid censoring scheme. Here the experiment is terminated at the time $T_{1}^{*}=min\{x_{m:n},T\}$, where $T\epsilon(0,\infty)$ and $1\leq m\leq n$ are prefixed and $x_{m:n}$ represents the $m^{th}$ failure time when $n$ units are on the experiment. This scheme also referred to as Type-I hybrid censoring scheme. Childs et al.\cite{childs2003exact} introduced a new censoring scheme known as, Type-II hybrid censoring scheme, in which the stopping time is $T_{2}^{*}=max\{x_{m:n},T\}$. However, in recent years, hybrid censoring scheme has grown increasingly prevalent in reliability and life-testing investigations, these schemes lack the flexibility to allow units to be removed from the experiment before it reaches its termination.
 
Later, Cohen\cite{cohen1963progressively} introduced progressive censoring scheme. Progressive censoring permits the removal of experimental objects at various stages to preserve the total cost and time associated with the experiment. From the perspective of experimental design, this approach could be useful in terms of identifying an early appropriate censoring scheme. Nevertheless, while this notion is frequently made in the literature, it’s not really met in real-world trials as the investigator can alter the censoring numbers during the study regardless of the circumstances. As a result, having a model that supports this type of adaptation is critical. In exchange for the efficiency benefit, Burkschat\cite{burkschat2008optimality} points out that progressive censoring has a longer test time than the conventional Type-II censoring strategy. In order to overcome the drawback of progressive Type-II censoring scheme Kundu and Joarder introduced Type-II progressively hybrid censoring scheme having fixed experimental time.

Thereafter Ng et al.\cite{ng2009statistical} introduced a censoring scheme called adaptive Type-II progressive censoring scheme. In which the effective sample size $m$ is fixed in advance, and the progressive censoring scheme is provided. But the number of items progressively removed from the experiment upon failure may change during the experiment. If the experimental time exceeds a prefixed time $T$ but the number of observed failures does not reach $m$, the experiment can terminate as soon as possible by adjusting the number of items progressively removed from the experiment upon failure.

Most of the censoring schemes discussed above do not guarantee a minimum of $m$ failures and a time constraint simultaneously. As a result, we’ve come up with a new censoring scheme named, Type I-Type II (T1-T2) mixture censoring scheme that ensures a minimum of $m$ and a maximum of $n$ failures with a prefixed supplementary time $S$. In this work, an experiment is carried out across $n$ units with the assumptions that the effective sample size $m$ and supplementary time $S$ are fixed prior to the experiment. At the time of $m$th failure the experiment is further continued for a prefixed supplementary time of $S$. The test is terminated at a time $T^{*}=min\{X_{n:n}, X_{m:n}+S\}$. For $S=0$, the proposed method converge to the conventional T2 censoring method.

The remaining sections of the paper are arranged as follows. In the second section, a description of the introduced censoring scheme using Weibull distribution is given. The third and fourth sections are the attempts to derive expressions for the expected duration of the experiment and the Fisher information matrix, respectively. The parameter estimation using maximum likelihood method (MLE) and Bayes method are given in Sections 5 and 6, respectively. A simulation study is performed in Section 7. In Section 8, real data is used to demonstrate the findings, and finally the paper is concluded in Section 9.

\section{Model Description}

Suppose n units are put on a life-testing experiment and let $X_{1},X_{2},X_{3},…,X_{n}$ be their corresponding lifetimes. If $r \, (m\leq r\leq n)$ denotes the number of failures occur before $T^{*}$, then the available data will be of the form;
\begin{align*}
&\text{Case I }: \{X_{1:n}, X_{2:n},…, X_{n:n}\} \quad \text{if} \quad T^{*}=X_{n:n} \\
&\text{Case II}: \{X_{1:n}, X_{2:n},…, X_{r:n}\} \quad \text{if} \quad T^{*}= X_{m:n}+S.
\end{align*}
The forms of likelihood function corresponding to both cases are respectively,
\begin{equation*}
	L(x_{i:n};\theta) \propto \left\{ \begin{array}{lll}
		\prod_{i=1}^{n}f(x_{i:n};\theta) & if & T^{*}=X_{n:n}\\
		\prod_{i=1}^{r}f(x_{i:n};\theta) \left\{1-F(X_{m:n}+S)\right\}^{n-r} & if & T^{*}= X_{m:n}+S
	\end{array}
\right.
\end{equation*}

In the fields of life-testing and reliability, the weibull distribution is one of the most often used lifetime model. Based on the above expression, the log-likelihood function for the weibull distribution with density function 
\begin{equation}
f(x;\gamma,\delta)=\gamma \delta x^{\gamma-1}e^{-\delta x^{\gamma}}; \, x,\delta, \gamma>0,
\end{equation} and cumulative distribution function 
\begin{equation}
F(x;\gamma,\delta)=1-e^{-\delta x^{\gamma}}; \, x,\delta, \gamma>0,
\end{equation} is obtained as,

\begin{equation*}
	\resizebox{1\hsize}{!}{%
	$\log L(\theta) \propto \Bigg\{ \begin{array}{lll}
		n \log \gamma+ n \log \delta + (\gamma-1)\sum_{i=1}^{n}\log x_{i}-\delta \sum_{i=1}^{n} x_{i}^{\gamma} & ; & r=n\\ \\
r \log \gamma+r \log \delta+ (\gamma-1)\sum_{i=1}^{r}\log x_{i}-
\delta\sum_{i=1}^{r}x_{i}^{\gamma}- \delta(n-r)U^{\gamma} & ; & r=m, m+1,...,n-1,
	\end{array}$%
}
\end{equation*}
where $U=X_{m:n}+S$.

\section{Expected Duration of Experiment}
Let $K$ be the number of failures observed and $T^{*}$ denotes the duration of experiment under T1-T2 censoring scheme. Clearly $K$ and $T^{*}$ are random variables.\\
The distribution of $K$ is given by,

\begin{align*}
	P(K=r)=&\binom{n}{r}F(T^{*};\theta)^{r}(1-F(T^{*};\theta))^{n-r}; \quad m\leq r\leq n.
\end{align*}

Therefore, the expected value of $K$ is,
\begin{align*}
	E[K]=\sum_{r=m}^{n}r\binom{n}{r}F(T^{*};\theta)^{r}(1-F(T^{*};\theta))^{n-r}
\end{align*}

Now consider the expected value of the duration of experiment, 
\begin{align*}
	E[T^{*}]=&E[\min\{X_{n:n},X_{m:n}+S\}]\\
	=&E\left[X_{m:n}+S|X_{m:n}+S\leq X_{n:n}\right]P\left[X_{m:n}+S\leq X_{n:n}\right]+\\
	&E\left[X_{n:n}|X_{n:n}<X_{m:n}+S\right]P\left[X_{n:n}<X_{m:n}+S\right]
\end{align*}

One additional way of determining the expected duration of experiment is as follows.

\begin{align*}
	E[T^{*}]=&E[min\{X_{n:n},X_{m:n}+S\}]\\
	=&\int_{0}^{\infty}P(min\{X_{n:n},X_{m:n}+S\}\geq x)dx\\
	=&\int_{0}^{\infty}P(X_{n:n}\geq x)P(X_{m:n}+S\geq x) dx\\
	=&\int_{0}^{\infty}(1-F_{n:n}(x;\theta))(1-F_{m:n}(x+S))dx
\end{align*}

\begin{lemma}
	$E[K]$ is invariant under scalar multiplication of $T$.
\end{lemma}
\begin{proof}
	Consider $T^{'}=\alpha T$, where $\alpha$ is the multiplication constant. Then $T^{*}$ changes to some $T^{**}$ and
	$P[T^{'}\leq T^{**}]=P[\alpha T\leq \alpha T^{*}]=P[T\leq T^{*}]=F(T^{*},\theta)$.
\end{proof}

\begin{lemma}
	Assume that $T^{+}$ represents the censoring time obtained by multiplying some constant $\alpha$ to $T$. Then, $E[T^{+}]=\alpha E[T^{*}]$.
\end{lemma}
\begin{proof}
	\begin{align*}
		E[T^{+}]=&E[min\{\alpha X_{n:n},\alpha(X_{m:n}+S)\}]\\
		=&E[\alpha(X_{m:n}+S)\;|\;\alpha(X_{m:n}+S)\leq\alpha X_{n:n}] \; P[\alpha(X_{m:n}+S)\leq \alpha X_{n:n}]+\\
		&E[\alpha X_{n:n}\;|\;\alpha X_{n:n}<\alpha(X_{m:n}+S)] \; P[\alpha X_{n:n}<\alpha(X_{m:n}+S)]\\
		=&\alpha E[X_{m:n}+S\;|\;X_{m:n}+S\leq X_{n:n}] \; P[X_{m:n}+S\leq X_{n:n}]+\\
		&E[X_{n:n}\;|\; X_{n:n}<X_{m:n}+S] \; P[X_{n:n}<X_{m:n}+S]\\
		=&\alpha E[T^{*}]
	\end{align*}
\end{proof}

\section{Fisher Information}
Effron and Johnstone \cite{efron1990fisher} shown that the Fisher information about its parameters of a continuous random variable can be expressed in terms of hazard function of the random variable. Park and Balakrishnan \cite{park2009simple} derived Fisher information in different hybrid censoring schemes. Based on this, we can derive Fisher information in the T1T2 censoring scheme.

\[ I(\gamma,\delta)= 	\left[ \begin{array}{cc}
	I_{11} & I_{12} \\
	I_{21} & I_{22}
\end{array}\right]
\]
where
\begin{align*}
	I_{11}=&\int_{0}^{\infty}\biggl\{\frac{\partial}{\partial\gamma}\log h(t)\biggr\}^{T} \biggl\{\frac{\partial}{\partial\gamma}\log h(t)\biggr\}\sum_{i=1}^{r}f_{i:n}(t;\gamma,\delta)dt\\
	=&\int_{0}^{\infty}\biggl\{\frac{1}{\gamma}+\log t\biggr\}^{2}\sum_{i=1}^{r}f_{i:n}(t;\gamma,\delta)dt\\
	I_{12}=&\int_{0}^{\infty}\biggl\{\frac{\partial}{\partial\gamma}\log h(t)\biggr\}^{T} \biggl\{\frac{\partial}{\partial\delta}\log h(t)\biggr\}\sum_{i=1}^{r}f_{i:n}(t;\gamma,\delta)dt\\
	=&\frac{1}{\delta}\int_{0}^{\infty}\biggl\{\frac{1}{\gamma}+\log t\biggr\}\sum_{i=1}^{r}f_{i:n}(t;\gamma,\delta)dt\\
	=&I_{21}\\
	I_{22}=&\int_{0}^{\infty}\biggl\{\frac{\partial}{\partial\delta}\log h(t)\biggr\}^{T}\biggl\{\frac{\partial}{\partial\delta}\log h(t)\biggr\}\sum_{i=1}^{r}f_{i:n}(t;\gamma,\delta)dt\\
	=&\int_{0}^{\infty}\left(\frac{1}{\delta}\right)^{2}\sum_{i=1}^{r}f_{i:n}(t;\gamma,\delta)dt\\
\text{and} \quad	f_{i:n}(t;\gamma,\delta)=&i\binom{n}{i}\gamma\delta t^{\gamma-1}e^{-\delta t^{\gamma}(n-i+1)}\left(1-e^{-\delta t^{\gamma}}\right)^{i-1}.
\end{align*}

\section{Maximum Likelihood Estimation}

The derivatives with respect to $\gamma$ and $\delta$ are respectively given by,

\begin{equation*}
	\resizebox{1.1\hsize}{!}{%
		$\frac{\partial \log L}{\partial \gamma} = \Bigg\{ \begin{array}{lll}
			\frac{n}{\gamma}+\sum_{i=1}^{n}\log x_{i}-\delta \sum_{i=1}^{n}x_{i}^{\gamma}\log x_{i} & ; & r=n\\ \\
			\frac{r}{\gamma}+\sum_{i=1}^{r}\log x_{i}-\delta\sum_{i=1}^{r}x_{i}^{\gamma}\log x_{i}-\delta(n-r)U^{\gamma}\log U & ; & r=m, m+1,...,n-1.
		\end{array}$%
	}
\end{equation*}

\begin{equation*}
		\frac{\partial \log L}{\partial \delta} = \Bigg\{ \begin{array}{lll}
			\frac{n}{\delta}-\sum_{i=1}^{n}x_{i}^{\gamma} & ; & r=n\\ \\
			\frac{r}{\delta}-\sum_{i=1}^{r}x_{i}^{\gamma}-(n-r)U^{\gamma} & ; & r=m, m+1,...,n-1.
		\end{array}
\end{equation*}

Now, combining Case I and Case II and equating to zero, the above equations become,

\begin{align}
	\frac{\partial\log L}{\partial \gamma}&=\frac{w}{\gamma}+\sum_{i=1}^{w}\log x_{i}-\delta P=0\\
	\frac{\partial\log L}{\partial\delta}&=\frac{w}{\delta}-Q=0
\end{align}
where $w=n$, $P=\sum_{i=1}^{n}x_{i}^{\gamma}\log x_{i}$ and $Q=\sum_{i=1}^{n}x_{i}^{\gamma}$ for case I and $w=r$, $P=\sum_{i=1}^{r}x_{i}^{\gamma}\log x_{i}+(n-r)U^{\gamma}\log U$ and $Q=\sum_{i=1}^{r}x_{i}^{\gamma}+(n-r)U^{\gamma}$ for case II.

From (2), we get
\begin{equation}
	\hat{\delta}=\frac{w}{Q}=v(\gamma), \text{(say)}.
\end{equation}
So that
\begin{equation}
	\hat{\gamma}=\frac{w}{v(\gamma)P-\sum_{i=1}^{w}\log x_{i}}=u(\gamma)
\end{equation}
where $v(\gamma)$ and $u(\gamma)$ are represents the functions of $\gamma$.

The MLEs, $\hat{\gamma}$ and $\hat{\delta}$ can be obtained by an iterative scheme proposed by Kundu \cite{kundu2007hybrid}.

\section{Bayes Estimation}
In this part, we investigate the Bayes estimates of the unknown parameters $\gamma$ and $\delta$ while using squared error loss function under T1-T2 censoring scheme. Here we assume that the prior distributions of $\gamma$ and $\delta$ are follows independent $Gamma(\alpha_{1},\beta_{1})$ and $Gamma(\alpha_{2},\beta_{2})$, respectively, where $\alpha_{1}, \alpha_{2}, \beta_{1}, \& \beta_{2}$ are the hyper parameters. When $\gamma$ and $\delta$ are both unknown, their conjugate prior won't exist. Therefore, gamma priors can be used in these circumstances because they are quite flexible and also take into account noninformative priors. 
Thus the independent prior distributions of $\gamma$ and $\delta$ are,

\begin{equation*}
	\pi_{1}(\gamma)\propto\gamma^{\alpha_{1}-1}e^{-\beta_{1}\gamma}; \, \gamma\geq0, \, \alpha_{1}, \beta_{1}>0,
\end{equation*}
\begin{equation*}
	\pi_{2}(\delta)\propto\delta^{\alpha_{2}-1}e^{-\beta_{2}\delta}; \, \delta\geq0, \alpha_{2}, \beta_{2}>0.
\end{equation*}

The joint prior distribution of the parameters $\gamma$ and $\delta$ is,

\begin{equation}
	\pi(\gamma,\delta)\propto\gamma^{\alpha_{1}-1}\delta^{\alpha_{2}-1}e^{-\beta_{1}\gamma-\beta_{2}\delta}; \, \gamma, \delta\geq0, \alpha_{1}, \alpha_{2}, \beta_{1}, \beta_{2}>0
\end{equation}

Using Bayes Theorem, the joint posterior density function of $\gamma$ and $\delta$ is given by
\begin{equation*}
	\pi(\gamma,\delta|\underline{X})\propto l(\gamma,\delta|\underline{X})\pi(\gamma,\delta)
\end{equation*}
\begin{equation}
	\pi(\gamma,\delta|\underline{X})\propto \gamma^{\alpha_{1}+w-1}\delta^{\alpha_{2}+w-1}\prod_{i=1}^{w}x_{i}^{\gamma-1}e^{-\delta\left[\beta_{2}+\sum_{i=1}^{w}x_{i}^{\gamma}+(n-w)U^{\gamma}\right]-\beta_{1}\gamma}
\end{equation}

To determine the Bayes estimates, we take into account both symmetric and asymmetric loss functions. For the symmetric loss function, the squared error loss function $(L_{SE})$ is considered. The $L_{SE}$ can be defined as,
\begin{equation}
	L_{SE}\left(g(\theta),\hat{g}(\theta)\right)=\left(g(\theta)-\hat{g}(\theta)\right)^{2},
\end{equation}
where $\hat{g}(\theta)$ is the estimate of $g(\theta)$.\\
For the asymmetric loss function, the LINEX loss function $(L_{LL})$ is considered. The $L_{LL}$ can be defined as,
\begin{equation}
	L_{LL}\big(g(\theta),\hat{g}(\theta)\big)=e^{d\left(g(\theta)-\hat{g}(\theta)\right)}-d\left(g(\theta)-\hat{g}(\theta)\right)-1, \quad d\neq0,
\end{equation}
where $d$ is the loss parameter. For the LINEX loss function, the Bayes estimate $\hat{\theta}_{LL}$ can be calculated as, 
\begin{equation}
	\hat{\theta}_{LL}=-\frac{1}{d}\log\left\{E_{\theta}\left(e^{-d\theta}\right)\right\},
\end{equation}
where the expectation exist and finite.

The conditional distribution of the unknown parameters,

\begin{equation}
	\pi(\gamma|\delta,\underline{X})\propto\gamma^{\alpha_{1}+w-1}\prod_{i=1}^{w}x_{i}^{\gamma-1}e^{-\delta\left[\sum_{i=1}^{w}x_{i}^{\gamma}+(n-w)U^{\gamma}\right]-\beta_{1}\gamma}
\end{equation}

\begin{equation}
	\pi(\delta|\gamma,\underline{X})\propto \text{Gamma}\left(\alpha_{2}+w,\beta_{2}+\sum_{i=1}^{w}x_{i}^{\gamma}+(n-w)U^{\gamma}\right)
\end{equation}
As the conditional distribution of $\gamma$ lacks an explicit form, we need to utilize the Metropolis-Hastings (M-H) technique to calculate Bayes estimates and to construct HPD credible intervals.

\subsection{M-H Algorithm}
The M-H algorithm proposed by Metropolis et al. \cite{metropolis1953equation} and Hasting \cite{hastings1970monte} is used to generate Bayes estimates of $\gamma$ and $\delta$. This method can be extremely helpful for constructing posterior samples using random proposal distributions when the posterior density is operationally infeasible. Based on the constructed posterior samples, Bayesian inference can be drawn regarding the unknown parameters $\gamma$ and $\delta$. The following steps can be used to construct the posterior samples.\\

STEPS\\
\textbf{Step 1:} Choose an initial value $(\gamma^{0},\delta^{0})$, and let $\sigma=\gamma^{i-1}$.\\
\textbf{Step 2:} Proposing a new point $\gamma^{i}$ from $\pi_{1}(\gamma|\delta^{0},\underline{X})$ from the proposing distribution $g(\gamma)\equiv N(\gamma^{0}1), \gamma>0$.\\
\textbf{Step 3:} Generate a random sample $u$ from Uniform(0,1) and a random sample $\psi$ from the proposal distribution. If $u\leq\frac{\pi(\psi)g(\sigma)}{\pi(\sigma)g(\psi)}$, then set $\gamma=\psi$. Here, the acceptance and rejection probabilities of $\psi$ are $min\left\{1,\frac{\pi(\psi)g(\sigma)}{\pi(\sigma)g(\psi)}\right\}$ and $1-min\left\{1,\frac{\pi(\psi)g(\sigma)}{\pi(\sigma)g(\psi)}\right\}$, respectively.\\
\textbf{Step 4:} Generate $\delta^{i}$'s from Gamma$\left(\alpha_{2}+w,\beta_{2}+\sum_{i=1}^{w}x_{i}^{\gamma}+(n-w)U^{\gamma}\right)$, as in step 3.\\
\textbf{Step 5:} Repeat steps 2 to 4, N times and obtained $(\gamma_{1},\delta_{1}), (\gamma_{2},\delta_{2}),..., (\gamma_{N},\delta_{N})$.\\

To ensure convergence and eliminate the impact of initial value selection, the first $M$ estimates are excluded. The remaining sample, consisting of $\gamma^{i}$ and $\delta^{i}$ with $i$ ranging from $M+1$ to $N$, where $N$ is sufficiently large, constitutes an estimated posterior sample that can be utilized for Bayesian inference. The $100(1-\alpha)\%$ HPD credible interval of the unknown parameter $\gamma$ can be constructed from the ordered samples $(\gamma_{(M+1)}, \gamma_{(M+2)},...,\gamma_{(N)})$by selecting the smallest interval, and similarly for $\delta$.

\section{Simulation Study}
This section is dedicated to presenting experimental results, primarily focused on examining the behavior of various methods under different sample sizes and censoring schemes for various values of parameters of $f(x;\gamma,\delta)$ under 
T1-T2 censoring scheme. The statistical software R is used to perform the computations.
\par
A MCMC approach is used to generate samples for different values of $(n, m, S)$ in order to assess the performance of MLE. Without loss of generality, we choose two different values for the parametres, $(\gamma, \delta)$ = (1, 1) and (1.5, 2). The censoring schemes $(n,m)$ = (100, 100), (100, 90), (100, 85), (60, 60), (60, 55), (60, 50), (60, 45), (30, 30), (30, 25), (30, 20), (30, 15), (15, 15), (15, 12), (15, 10), and (15, 7) are chosen for the supplymentary times $S=0.1$ and $S=0.2$ with 1000 replications. Bias and mean squared error (MSE) are computed for the unknown parameters in order to examine the effectiveness of MLE, and the results are given in the Table \ref*{mle}. As can be seen from Table \ref{mle}, the MLE of parameters performs pretty satisfactorily with respect to bias and MSE. The MSE value is declining as sample size increases, indicating that the MLEs of the unknown parameters are consistent. Moreover, we can observe that the MSE value obtained for S equal to 0.2 is smaller compared to the MSE value obtained for S equal to 0.1. As the supplementary time increases, there is a greater chance of more failures occurring, which, in turn, leads to more precise statistical results.
The average confidence length (CL) and coverage probability (CP) based on $95 \%$ confidence interval are determined and the results are given in Table \ref{ci1} and \ref{ci2}. As the sample size grows, the value of CL gets shorter, and the CP remains at its nominal level.
\par
The Bayes estimates for squared error and the LINEX loss functions are derived using the M-H algorithm. Under LINEX, the arbitrarily chosen values of the loss parameter $d$ are -1 and 1. For the prior distribution, hyperparameters are chosen in such a way that the corresponding prior means are somewhat close to the actual parameter values. Hence, the selected values are $(\alpha_{1}=1, \beta_{1}=1, \alpha_{2}=1, \beta_{2}=1)$ and $(\alpha_{1}=2.25, \beta_{1}=1.5, \alpha_{2}=4, \beta_{2}=2)$ (see, Kundu \cite{kundu2008bayesian}). Tables \ref{bayes1}-\ref{bayes2} show the estimated bias and risk for various censoring schemes under the aforementioned loss functions. When we compare the bias and risk of each estimate, it is clear that the estimates obtained under the LINEX loss function perform better. The values of bias and risk are decreasing with the increase in the effective sample size. Table 6 displays the CL and CP based on the $95\%$ HPD credible interval (HPDC). As the effective sample size rises, the CL of HPDC becomes narrower. Meanwhile, the nominal level of the CP is maintained. Comparing CL of HPDC with MLE, it is clear that CL of HPDC performs better for small sample sizes while they both have values that are nearly equivalent for higher sample sizes.

\begin{table}[htbp]
	\centering
	\caption{Average values of the biases and MSEs for the maximum likelihood estimates of ($\gamma$, $\delta$) = (1, 1) and (1.5, 2) under different censoring schemes.}
	\resizebox{15cm}{!}{
		\renewcommand*{\arraystretch}{1.5}
	\begin{tabular}{c|cccc|cccc}
		\hline
		\multirow{2}[1]{*}{Censoring Scheme} & \multicolumn{4}{c|}{\multirow{2}[1]{*}{$\gamma$=1, $\delta$=1}} & \multicolumn{4}{c}{\multirow{2}[1]{*}{$\gamma$=1.5, $\delta$=2}} \\
		& \multicolumn{4}{c|}{}         & \multicolumn{4}{c}{} \\
		\cline{2-9}        & \multicolumn{1}{l}{Bias($\gamma$)} & \multicolumn{1}{l}{Bias($\delta$)} & \multicolumn{1}{l}{MSE($\gamma$)} & \multicolumn{1}{l|}{MSE($\delta$)} & \multicolumn{1}{l}{Bias($\gamma$)} & \multicolumn{1}{l}{Bias($\delta$)} & \multicolumn{1}{l}{MSE($\gamma$)} & \multicolumn{1}{l}{MSE($\delta$)} \\
		\hline
		    n=100, m=100, S=0.1 & 0.0101 & 0.0077 & 0.0064 & 0.0115 & 0.0151 & 0.0291 & 0.0144 & 0.0466 \\
		n=100, m=90, S=0.1 & 0.0112 & 0.0096 & 0.0077 & 0.0122 & 0.0152 & 0.0329 & 0.0168 & 0.0550 \\
		n=100, m=85, S=0.1 & 0.0122 & 0.0109 & 0.0083 & 0.0126 & 0.0169 & 0.0370 & 0.0178 & 0.0582 \\
		n=100, m=80, S=0.1 & 0.0125 & 0.0118 & 0.0089 & 0.0131 & 0.0174 & 0.0401 & 0.0188 & 0.0640 \\
		n=60, m=60, S=0.1 & 0.0189 & 0.0080 & 0.0111 & 0.0183 & 0.0284 & 0.0415 & 0.0251 & 0.0751 \\
		n=60, m=55, S=0.1 & 0.0219 & 0.0109 & 0.0135 & 0.0192 & 0.0306 & 0.0489 & 0.0294 & 0.0865 \\
		n=60, m=50, S=0.1 & 0.0264 & 0.0153 & 0.0161 & 0.0205 & 0.0365 & 0.0641 & 0.0333 & 0.1009 \\
		n=60, m=45, S=0.1 & 0.0306 & 0.0210 & 0.0192 & 0.0237 & 0.0397 & 0.0760 & 0.0386 & 0.1240 \\
		n=30, m=30, S=0.1 & 0.0472 & 0.0284 & 0.0274 & 0.0454 & 0.0709 & 0.1234 & 0.0616 & 0.2104 \\
		n=30, m=25, S=0.1 & 0.0577 & 0.0436 & 0.0387 & 0.0576 & 0.0814 & 0.1659 & 0.0799 & 0.3046 \\
		n=30, m=20, S=0.1 & 0.0805 & 0.0811 & 0.0568 & 0.0988 & 0.1008 & 0.2485 & 0.1069 & 0.5430 \\
		n=30, m=15, S=0.1 & 0.1105 & 0.1468 & 0.0853 & 0.1876 & 0.1388 & 0.4397 & 0.1579 & 1.3586 \\
		n=15, m=15, S=0.1 & 0.0992 & 0.0558 & 0.0722 & 0.1133 & 0.1488 & 0.2633 & 0.1625 & 0.6763 \\
		n=15, m=12, S=0.1 & 0.1294 & 0.1062 & 0.1107 & 0.1910 & 0.1739 & 0.3879 & 0.2191 & 1.2546 \\
		n=15, m=10, S=0.1 & 0.1640 & 0.1719 & 0.1501 & 0.2932 & 0.2090 & 0.5749 & 0.2843 & 2.4203 \\
		n=15, m=7, S=0.1 & 0.2473 & 0.4116 & 0.2840 & 1.2957 & 0.2950 & 1.1548 & 0.4437 & 8.1042 \\
		n=100, m=100, S=0.2 & 0.0101 & 0.0077 & 0.0064 & 0.0115 & 0.0151 & 0.0291 & 0.0144 & 0.0466 \\
		n=100, m=90, S=0.2 & 0.0104 & 0.0090 & 0.0076 & 0.0120 & 0.0126 & 0.0281 & 0.0158 & 0.0524 \\
		n=100, m=85, S=0.2 & 0.0108 & 0.0097 & 0.0080 & 0.0124 & 0.0142 & 0.0313 & 0.0168 & 0.0548 \\
		n=100, m=80, S=0.2 & 0.0124 & 0.0114 & 0.0086 & 0.0128 & 0.0149 & 0.0333 & 0.0177 & 0.0585 \\
		n=60, m=60, S=0.2 & 0.0189 & 0.0080 & 0.0111 & 0.0183 & 0.0284 & 0.0415 & 0.0251 & 0.0751 \\
		n=60, m=55, S=0.2 & 0.0209 & 0.0101 & 0.0134 & 0.0191 & 0.0275 & 0.0426 & 0.0281 & 0.0817 \\
		n=60, m=50, S=0.2 & 0.0246 & 0.0139 & 0.0153 & 0.0202 & 0.0325 & 0.0537 & 0.0305 & 0.0904 \\
		n=60, m=45, S=0.2 & 0.0279 & 0.0178 & 0.0179 & 0.0222 & 0.0345 & 0.0620 & 0.0335 & 0.1050 \\
		n=30, m=30, S=0.2 & 0.0472 & 0.0284 & 0.0274 & 0.0454 & 0.0709 & 0.1234 & 0.0616 & 0.2104 \\
		n=30, m=25, S=0.2 & 0.0546 & 0.0399 & 0.0367 & 0.0542 & 0.0742 & 0.1435 & 0.0741 & 0.2639 \\
		n=30, m=20, S=0.2 & 0.0687 & 0.0619 & 0.0496 & 0.0748 & 0.0817 & 0.1811 & 0.0886 & 0.3714 \\
		n=30, m=15, S=0.2 & 0.0930 & 0.1071 & 0.0704 & 0.1271 & 0.1000 & 0.2584 & 0.1152 & 0.5804 \\
		n=15, m=15, S=0.2 & 0.0992 & 0.0558 & 0.0722 & 0.1133 & 0.1488 & 0.2633 & 0.1625 & 0.6763 \\
		n=15, m=12, S=0.2 & 0.1197 & 0.0896 & 0.1032 & 0.1608 & 0.1536 & 0.3106 & 0.1910 & 0.8809 \\
		n=15, m=10, S=0.2 & 0.1428 & 0.1323 & 0.1300 & 0.2247 & 0.1660 & 0.3807 & 0.2210 & 1.1855 \\
		n=15, m=7, S=0.2 & 0.1901 & 0.2315 & 0.1960 & 0.4099 & 0.2008 & 0.5661 & 0.2876 & 2.1479 \\
		\hline
	\end{tabular}%
}
	\label{mle}%
\end{table}%

\begin{table}[htbp]
	\centering
	\caption{Average values of the biases and Risks for the Bayesian estimates of parameters ($\gamma$, $\delta$) = (1, 1) under different censoring schemes.}
	\resizebox{15cm}{!}{
		\renewcommand*{\arraystretch}{1.5}
	\begin{tabular}{lcc|cc|cc|cc|cc|cc}
		\hline
		\multicolumn{1}{c}{\multirow{4}[7]{*}{$\gamma$=1, $\delta$=1}} & \multicolumn{4}{c}{SE}        & \multicolumn{8}{c}{LL} \\
		\cmidrule{2-13}          & \multicolumn{2}{c}{\multirow{2}[4]{*}{Bias}} & \multicolumn{2}{c}{\multirow{2}[4]{*}{Risk}} & \multicolumn{4}{c}{Bias}      & \multicolumn{4}{c}{Risk} \\
		\cmidrule{6-13}          & \multicolumn{2}{c}{} & \multicolumn{2}{c}{} & \multicolumn{2}{c}{v=-1} & \multicolumn{2}{c}{v=1} & \multicolumn{2}{c}{v=-1} & \multicolumn{2}{c}{v=1} \\
		\cmidrule{2-13}          & \multicolumn{1}{l}{$\gamma$} & \multicolumn{1}{l}{$\delta$} & \multicolumn{1}{l}{$\gamma$} & \multicolumn{1}{l}{$\delta$} & \multicolumn{1}{l}{$\gamma$} & \multicolumn{1}{l}{$\delta$} & \multicolumn{1}{l}{$\gamma$} & \multicolumn{1}{l}{$\delta$} & \multicolumn{1}{l}{$\gamma$} & \multicolumn{1}{l}{$\delta$} & \multicolumn{1}{l}{$\gamma$} & \multicolumn{1}{l}{$\delta$} \\
		\hline
		n=100, m=100, S=0.1 & -0.1807 & 0.0245 & 0.0402 & 0.0107 & 0.0038 & 0.0051 & -0.0037 & -0.0049 & 0.0185 & 0.0054 & 0.0219 & 0.0053 \\
		n=100, m=90, S=0.1 & -0.2151 & 0.0367 & 0.0544 & 0.0157 & 0.0041 & 0.0075 & -0.0041 & -0.0069 & 0.0248 & 0.0078 & 0.0299 & 0.0080 \\
		n=100, m=85, S=0.1 & -0.1917 & 0.0409 & 0.0457 & 0.0145 & 0.0044 & 0.0066 & -0.0045 & -0.0062 & 0.0210 & 0.0073 & 0.0249 & 0.0072 \\
		n=100, m=80, S=0.1 & -0.1923 & 0.0264 & 0.0458 & 0.0128 & 0.0044 & 0.0062 & -0.0045 & -0.0059 & 0.0210 & 0.0065 & 0.0251 & 0.0064 \\
		&       &       &       &       &       &       &       &       &       &       &       &  \\
		n=60, m=60, S=0.1 & -0.1385 & -0.0012 & 0.0270 & 0.0154 & 0.0039 & 0.0077 & -0.0040 & -0.0076 & 0.0126 & 0.0076 & 0.0145 & 0.0078 \\
		n=60, m=55, S=0.1 & -0.1854 & -0.0077 & 0.0498 & 0.0223 & 0.0077 & 0.0112 & -0.0077 & -0.0111 & 0.0226 & 0.0111 & 0.0276 & 0.0113 \\
		n=60, m=50, S=0.1 & -0.1912 & -0.0158 & 0.0551 & 0.0243 & 0.0094 & 0.0122 & -0.0092 & -0.0118 & 0.0248 & 0.0118 & 0.0309 & 0.0126 \\
		n=60, m=45, S=0.1 & -0.1458 & 0.0050 & 0.0406 & 0.0263 & 0.0098 & 0.0135 & -0.0095 & -0.0127 & 0.0184 & 0.0129 & 0.0225 & 0.0136 \\
		&       &       &       &       &       &       &       &       &       &       &       & \\
		n=30, m=30, S=0.1 & 0.0065 & 0.0269 & 0.0206 & 0.0409 & 0.0103 & 0.0203 & -0.0102 & -0.0200 & 0.0103 & 0.0211 & 0.0103 & 0.0203 \\
		n=30, m=25, S=0.1 & 0.0175 & 0.0774 & 0.0257 & 0.0404 & 0.0126 & 0.0175 & -0.0128 & -0.0170 & 0.0133 & 0.0216 & 0.0126 & 0.0192 \\
		n=30, m=20, S=0.1 & -0.1085 & -0.0158 & 0.0527 & 0.0525 & 0.0209 & 0.0273 & -0.0201 & -0.0251 & 0.0239 & 0.0252 & 0.0297 & 0.0282 \\
		n=30, m=15, S=0.1 & -0.2812 & -0.3702 & 0.1647 & 0.3088 & 0.0456 & 0.0999 & -0.0408 & -0.0768 & 0.0675 & 0.1160 & 0.1053 & 0.2299 \\
		&       &       &       &       &       &       &       &       &       &       &       &  \\
		n=15, m=15, S=0.1 & -0.1103 & 0.0076 & 0.0669 & 0.0664 & 0.0286 & 0.0355 & -0.0261 & -0.0313 & 0.0296 & 0.0321 & 0.0388 & 0.0359 \\
		n=15, m=12, S=0.1 & -0.1113 & -0.0013 & 0.0768 & 0.0744 & 0.0334 & 0.0392 & -0.0311 & -0.0355 & 0.0343 & 0.0361 & 0.0444 & 0.0400 \\
		n=15, m=10, S=0.1 & -0.1843 & -0.0774 & 0.1151 & 0.0930 & 0.0417 & 0.0446 & -0.0396 & -0.0425 & 0.0496 & 0.0431 & 0.0693 & 0.0524 \\
		n=15, m=7, S=0.1 & -0.2589 & -0.2617 & 0.1997 & 0.3329 & 0.0685 & 0.1587 & -0.0641 & -0.1144 & 0.0819 & 0.1247 & 0.1285 & 0.2608 \\
		&       &       &       &       &       &       &       &       &       &       &       &  \\
		n=100, m=100, S=0.2 & -0.1807 & 0.0245 & 0.0302 & 0.0107 & 0.0038 & 0.0051 & -0.0037 & -0.0049 & 0.0185 & 0.0054 & 0.0219 & 0.0053 \\
		n=100, m=90, S=0.2 & -0.2043 & 0.0181 & 0.0326 & 0.0136 & 0.0055 & 0.0069 & -0.0044 & -0.0064 & 0.0190 & 0.0067 & 0.0221 & 0.0069 \\
		n=100, m=85, S=0.2 & -0.1738 & 0.0333 & 0.0493 & 0.0147 & 0.0046 & 0.0070 & -0.0045 & -0.0066 & 0.0187 & 0.0075 & 0.0215 & 0.0073 \\
		n=100, m=80, S=0.2 & -0.1962 & 0.0307 & 0.0589 & 0.0147 & 0.0058 & 0.0071 & -0.0052 & -0.0067 & 0.0223 & 0.0074 & 0.0269 & 0.0074 \\
		&       &       &       &       &       &       &       &       &       &       &       &  \\
		n=60, m=60, S=0.2 & -0.1385 & -0.0012 & 0.0270 & 0.0154 & 0.0039 & 0.0077 & -0.0040 & -0.0076 & 0.0126 & 0.0076 & 0.0145 & 0.0078 \\
		n=60, m=55, S=0.2 & -0.1589 & -0.0064 & 0.0364 & 0.0188 & 0.0055 & 0.0095 & -0.0056 & -0.0093 & 0.0167 & 0.0093 & 0.0198 & 0.0096 \\
		n=60, m=50, S=0.2 & -0.1748 & -0.0191 & 0.0397 & 0.0244 & 0.0091 & 0.0117 & -0.0087 & -0.0113 & 0.0167 & 0.0116 & 0.0197 & 0.0130 \\
		n=60, m=45, S=0.2 & -0.1818 & 0.0203 & 0.0399 & 0.0259 & 0.0079 & 0.0129 & -0.0076 & -0.0125 & 0.0150 & 0.0130 & 0.0281 & 0.0130 \\
		&       &       &       &       &       &       &       &       &       &       &       &  \\
		n=30, m=30, S=0.2 & 0.0065 & 0.0269 & 0.0206 & 0.0409 & 0.0103 & 0.0103 & -0.0102 & -0.0110 & 0.0103 & 0.0211 & 0.0103 & 0.0203 \\
		n=30, m=25, S=0.2 & 0.0480 & 0.0891 & 0.0264 & 0.0450 & 0.0121 & 0.0186 & -0.0120 & -0.0185 & 0.0138 & 0.0244 & 0.0127 & 0.0210 \\
		n=30, m=20, S=0.2 & -0.1038 & -0.0157 & 0.0502 & 0.0429 & 0.0201 & 0.0221 & -0.0195 & -0.0206 & 0.0229 & 0.0206 & 0.0281 & 0.0228 \\
		n=30, m=15, S=0.2 & -0.1553 & -0.1185 & 0.0899 & 0.0976 & 0.0341 & 0.0426 & -0.0319 & -0.0407 & 0.0392 & 0.0437 & 0.0533 & 0.0563 \\
		&       &       &       &       &       &       &       &       &       &       &       &  \\
		n=15, m=15, S=0.2 & -0.1103 & 0.0076 & 0.0669 & 0.0664 & 0.0286 & 0.0355 & -0.0261 & -0.0313 & 0.0296 & 0.0321 & 0.0388 & 0.0359 \\
		n=15, m=12, S=0.2 & -0.1147 & 0.0149 & 0.0814 & 0.0690 & 0.0348 & 0.0365 & -0.0336 & -0.0328 & 0.0368 & 0.0339 & 0.0465 & 0.0368 \\
		n=15, m=10, S=0.2 & -0.1636 & -0.0467 & 0.1055 & 0.0775 & 0.0399 & 0.0388 & -0.0387 & -0.0366 & 0.0462 & 0.0366 & 0.0620 & 0.0426 \\
		n=15, m=7, S=0.2 & -0.1570 & -0.0917 & 0.1355 & 0.1877 & 0.0585 & 0.0997 & -0.0532 & -0.0819 & 0.0584 & 0.0819 & 0.0835 & 0.1193 \\
		\hline
	\end{tabular}%
}
	\label{bayes1}%
\end{table}%

\begin{table}[htbp]
	\centering
	\caption{Average values of the biases and Risks for the Bayesian estimates of parameters ($\gamma$, $\delta$) = (1.5, 2) under different censoring schemes.}
		\resizebox{15cm}{!}{
		\renewcommand*{\arraystretch}{1.5}
	\begin{tabular}{lcc|cc|cc|cc|cc|cc}
		\hline
		\multicolumn{1}{c}{\multirow{4}[7]{*}{$\gamma$=1.5, $\delta$=2}} & \multicolumn{4}{c}{SE}        & \multicolumn{8}{c}{LL} \\
		\cmidrule{2-13}          & \multicolumn{2}{c}{\multirow{2}[4]{*}{Bias}} & \multicolumn{2}{c}{\multirow{2}[4]{*}{Risk}} & \multicolumn{4}{c}{Bias}      & \multicolumn{4}{c}{Risk} \\
		\cmidrule{6-13}          & \multicolumn{2}{c}{} & \multicolumn{2}{c}{} & \multicolumn{2}{c}{v=-1} & \multicolumn{2}{c}{v=1} & \multicolumn{2}{c}{v=-1} & \multicolumn{2}{c}{v=1} \\
		\cmidrule{2-13}          & \multicolumn{1}{l}{$\gamma$} & \multicolumn{1}{l}{$\delta$} & \multicolumn{1}{l}{$\gamma$} & \multicolumn{1}{l}{$\delta$} & \multicolumn{1}{l}{$\gamma$} & \multicolumn{1}{l}{$\delta$} & \multicolumn{1}{l}{$\gamma$} & \multicolumn{1}{l}{$\delta$} & \multicolumn{1}{l}{$\gamma$} & \multicolumn{1}{l}{$\delta$} & \multicolumn{1}{l}{$\gamma$} & \multicolumn{1}{l}{$\delta$} \\
		\hline
		n=100, m=100, S=0.1 & -0.3000 & -0.2028 & 0.1100 & 0.0931 & 0.0100 & 0.0265 & -0.0100 & -0.0255 & 0.0483 & 0.0404 & 0.0635 & 0.0549 \\
		n=100, m=90, S=0.1 & -0.2921 & -0.2416 & 0.1045 & 0.1167 & 0.0098 & 0.0297 & -0.0095 & -0.0287 & 0.0459 & 0.0499 & 0.0603 & 0.0700 \\
		n=100, m=85, S=0.1 & -0.2543 & -0.1387 & 0.0880 & 0.0796 & 0.0116 & 0.0303 & -0.0117 & -0.0300 & 0.0389 & 0.0357 & 0.0503 & 0.0455 \\
		n=100, m=80, S=0.1 & -0.2664 & -0.1790 & 0.0962 & 0.0828 & 0.0128 & 0.0265 & -0.0124 & -0.0245 & 0.0421 & 0.0358 & 0.0556 & 0.0491 \\
		&       &       &       &       &       &       &       &       &       &       &       &  \\
		n=60, m=60, S=0.1 & -0.2141 & -0.1919 & 0.0733 & 0.1100 & 0.0138 & 0.0381 & -0.0136 & -0.0353 & 0.0325 & 0.0470 & 0.0419 & 0.0667 \\
		n=60, m=55, S=0.1 & -0.2368 & -0.2132 & 0.0895 & 0.1295 & 0.0170 & 0.0446 & -0.0165 & -0.0399 & 0.0391 & 0.0541 & 0.0521 & 0.0809 \\
		n=60, m=50, S=0.1 & -0.2476 & -0.2585 & 0.1019 & 0.1762 & 0.0208 & 0.0573 & -0.0197 & -0.0526 & 0.0438 & 0.0724 & 0.0603 & 0.1128 \\
		n=60, m=45, S=0.1 & -0.2355 & -0.2172 & 0.0947 & 0.1581 & 0.0200 & 0.0594 & -0.0192 & -0.0524 & 0.0410 & 0.0653 & 0.0556 & 0.1015 \\
		&       &       &       &       &       &       &       &       &       &       &       &  \\
		n=30, m=30, S=0.1 & 0.0164 & 0.1118 & 0.0410 & 0.1105 & 0.0209 & 0.0507 & -0.0199 & -0.0475 & 0.0205 & 0.0608 & 0.0209 & 0.0526 \\
		n=30, m=25, S=0.1 & -0.0041 & 0.1723 & 0.0545 & 0.1293 & 0.0277 & 0.0510 & -0.0267 & -0.0485 & 0.0270 & 0.0748 & 0.0283 & 0.0581 \\
		n=30, m=20, S=0.1 & -0.0830 & -0.0682 & 0.0763 & 0.2182 & 0.0363 & 0.1183 & -0.0331 & -0.0975 & 0.0344 & 0.0980 & 0.0437 & 0.1368 \\
		n=30, m=15, S=0.1 & -0.1491 & -0.2686 & 0.0986 & 0.4535 & 0.0388 & 0.2268 & -0.0375 & -0.1657 & 0.0435 & 0.1709 & 0.0576 & 0.3726 \\
		&       &       &       &       &       &       &       &       &       &       &       &  \\
		n=15, m=15, S=0.1 & -0.1570 & -0.0601 & 0.1141 & 0.2312 & 0.0463 & 0.1195 & -0.0433 & -0.1079 & 0.0495 & 0.1091 & 0.0684 & 0.1367 \\
		n=15, m=12, S=0.1 & -0.0868 & 0.0384 & 0.1218 & 0.2507 & 0.0584 & 0.1324 & -0.0560 & -0.1168 & 0.0564 & 0.1294 & 0.0694 & 0.1369 \\
		n=15, m=10, S=0.1 & -0.1774 & -0.1020 & 0.1865 & 0.3356 & 0.0838 & 0.1788 & -0.0725 & -0.1489 & 0.0778 & 0.1500 & 0.1211 & 0.2222 \\
		n=15, m=7, S=0.1 & -0.2125 & -0.3120 & 0.2017 & 0.6883 & 0.0800 & 0.3317 & -0.0759 & -0.2579 & 0.0849 & 0.2593 & 0.1273 & 0.5915 \\
		&       &       &       &       &       &       &       &       &       &       &       &  \\
		n=100, m=100, S=0.2 & -0.3000 & -0.2028 & 0.1100 & 0.0931 & 0.0100 & 0.0265 & -0.0100 & -0.0255 & 0.0483 & 0.0404 & 0.0635 & 0.0549 \\
		n=100, m=90, S=0.2 & -0.3016 & -0.2259 & 0.1107 & 0.1039 & 0.0099 & 0.0268 & -0.0098 & -0.0262 & 0.0485 & 0.0449 & 0.0639 & 0.0616 \\
		n=100, m=85, S=0.2 & -0.2929 & -0.2453 & 0.1165 & 0.1166 & 0.0106 & 0.0280 & -0.0101 & -0.0283 & 0.0466 & 0.0502 & 0.0617 & 0.0690 \\
		n=100, m=80, S=0.2 & -0.3172 & -0.2655 & 0.1253 & 0.1282 & 0.0124 & 0.0295 & -0.0123 & -0.0282 & 0.0544 & 0.0543 & 0.0733 & 0.0776 \\
		&       &       &       &       &       &       &       &       &       &       &       &  \\
		n=60, m=60, S=0.2 & -0.2141 & -0.1919 & 0.0733 & 0.1100 & 0.0138 & 0.0381 & -0.0136 & -0.0353 & 0.0325 & 0.0470 & 0.0419 & 0.0667 \\
		n=60, m=55, S=0.2 & -0.2224 & -0.2001 & 0.0805 & 0.1198 & 0.0156 & 0.0416 & -0.0155 & -0.0384 & 0.0355 & 0.0508 & 0.0463 & 0.0733 \\
		n=60, m=50, S=0.2 & -0.2549 & -0.2446 & 0.1045 & 0.1472 & 0.0208 & 0.0462 & -0.0189 & -0.0416 & 0.0447 & 0.0609 & 0.0625 & 0.0929 \\
		n=60, m=45, S=0.2 & -0.2654 & -0.2651 & 0.1122 & 0.1742 & 0.0213 & 0.0535 & -0.0204 & -0.0507 & 0.0481 & 0.0721 & 0.0666 & 0.1101 \\
		&       &       &       &       &       &       &       &       &       &       &       &  \\
		n=30, m=30, S=0.2 & 0.0164 & 0.1118 & 0.0410 & 0.1105 & 0.0209 & 0.0507 & -0.0199 & -0.0475 & 0.0205 & 0.0608 & 0.0209 & 0.0526 \\
		n=30, m=25, S=0.2 & 0.0215 & 0.1535 & 0.0479 & 0.1240 & 0.0244 & 0.0517 & -0.0230 & -0.0486 & 0.0240 & 0.0705 & 0.0244 & 0.0567 \\
		n=30, m=20, S=0.2 & -0.0673 & 0.0020 & 0.0682 & 0.1549 & 0.0334 & 0.0807 & -0.0304 & -0.0741 & 0.0311 & 0.0771 & 0.0387 & 0.0838 \\
		n=30, m=15, S=0.2 & -0.0398 & 0.0052 & 0.0717 & 0.2535 & 0.0361 & 0.1389 & -0.0342 & -0.1153 & 0.0342 & 0.1229 & 0.0390 & 0.1482 \\
		&       &       &       &       &       &       &       &       &       &       &       &  \\
		n=15, m=15, S=0.2 & -0.1570 & -0.0601 & 0.1141 & 0.2312 & 0.0463 & 0.1195 & -0.0433 & -0.1079 & 0.0495 & 0.1091 & 0.0684 & 0.1367 \\
		n=15, m=12, S=0.2 & -0.0064 & 0.1491 & 0.0910 & 0.2526 & 0.0457 & 0.1221 & -0.0452 & -0.1086 & 0.0460 & 0.1449 & 0.0471 & 0.1225 \\
		n=15, m=10, S=0.2 & -0.1921 & -0.1076 & 0.1435 & 0.2756 & 0.0635 & 0.1424 & -0.0549 & -0.1224 & 0.0639 & 0.1225 & 0.0892 & 0.1764 \\
		n=15, m=7, S=0.2 & -0.0583 & 0.0916 & 0.1439 & 0.3993 & 0.0715 & 0.2239 & -0.0688 & -0.1702 & 0.0689 & 0.2077 & 0.0903 & 0.2330 \\
		\hline
	\end{tabular}%
}
	\label{bayes2}%
\end{table}%

\begin{table}[htbp]
	\centering
	\caption{Confidence length and coverage probability for the parameters ($\gamma$,$\delta$) under different censoring schemes.}
		\resizebox{15cm}{!}{
		\renewcommand*{\arraystretch}{1.3}
	\begin{tabular}{ccc|cc|cc|cc}
		\hline
		& \multicolumn{4}{c|}{$\gamma$=1, $\delta$=1} & \multicolumn{4}{c}{$\gamma$=1.5, $\delta$=2} \\ \hline
		& \multicolumn{2}{c|}{Bayesian HPDC} & \multicolumn{2}{c|}{MLE} & \multicolumn{2}{c|}{Bayesian HPDC} & \multicolumn{2}{c}{MLE} \\
		\hline
	        & \multicolumn{1}{c}{CL} & \multicolumn{1}{c|}{CP} & \multicolumn{1}{c}{CL} & \multicolumn{1}{c|}{CP} & \multicolumn{1}{c}{CL} & \multicolumn{1}{c|}{CP} & \multicolumn{1}{c}{CL} & \multicolumn{1}{c}{CP} \\
		\cline{2-9}    \multicolumn{1}{c}{\multirow{2}[1]{*}{n=100, m=100, S=0.1}} & (0.3180, & (0.951, & (0.3099, & \multicolumn{1}{c|}{(0.952,} & (0.5288, & (0.952, & (0.4649, & (0.952, \\
		& 0.3563) & 0.952) & 0.4159) & 0.96)  & 0.8520) & 0.952) & 0.8187) & 0.957) \\
		\multirow{2}[0]{*}{n=100, m=90, S=0.1} & (0.3688, & (0.952, & (0.3442, & \multicolumn{1}{c|}{(0.946,} & (0.5515, & (0.951, & (0.5068, & (0.95, \\
		& 0.4162) & 0.951) & 0.4242) & 0.954) & 0.9159) & 0.951) & 0.8822) & 0.947) \\
		\multirow{2}[0]{*}{n=100, m=85, S=0.1} & (0.3686, & (0.953, & (0.3601, & (0.954, & (0.5738, & (0.952, & (0.5270, & (0.947, \\
		& 0.4303) & 0.953) & 0.4315) & 0.956) & 0.9019) & 0.951) & 0.9228) & 0.957) \\
		\multirow{2}[0]{*}{n=100, m=80, S=0.1} & (0.3377, & (0.952, & (0.3758, & (0.952, & (0.5815, & (0.951, & (0.5471, & (0.95, \\
		& 0.4366) & 0.951) & 0.4410) & 0.955) & 0.8816) & 0.954) & 0.9688) & 0.956) \\
		\hline
		\multirow{2}[0]{*}{n=60, m=60, S=0.1} & (0.3242, & (0.957, & (0.4038, & (0.954, & (0.6451, & (0.951, & (0.6056, & (0.954, \\
		& 0.4537) & 0.952) & 0.5375) & 0.948) & 1.0304) & 0.953) & 1.0671) & 0.959) \\
		\multirow{2}[0]{*}{n=60, m=55, S=0.1} & (0.4955, & (0.953, & (0.4425, & (0.954, & (0.7269, & (0.952, & (0.6526, & (0.955, \\
		& 0.5992) & 0.955) & 0.5468) & 0.953) & 1.0754) & 0.951) & 1.1389) & 0.956) \\
		\multirow{2}[0]{*}{n=60, m=50, S=0.1} & (0.4951, & (0.954, & (0.4782, & (0.951, & (0.7693, & (0.951, & (0.6979, & (0.959, \\
		& 0.6076) & 0.954) & 0.5642) & 0.952) & 1.3370) & 0.952) & 1.2349) & 0.962) \\
		\multirow{2}[0]{*}{n=60, m=45, S=0.1} & (0.5466, & (0.954, & (0.5151, & (0.956, & (0.7519, & (0.951, & (0.7445, & (0.955, \\
		& 0.6179) & 0.951) & 0.5921) & 0.961) & 1.2934) & \multicolumn{1}{c|}{0.951)} & 1.3555) & 0.966) \\
		\hline
		\multirow{2}[0]{*}{n=30, m=30, S=0.1} & (0.5174, & (0.953, & (0.5891, & (0.947, & (0.8079, & (0.951, & (0.8836, & (0.947, \\
		& 0.7871) & 0.951) & 0.7750) & 0.943) & 1.1309) & 0.951) & 1.5953) & 0.966) \\
		\multirow{2}[0]{*}{n=30, m=25, S=0.1} & (0.5806, & (0.951, & (0.6969, & (0.957, & (0.8977, & (0.951, & (1.0151, & (0.958, \\
		& 0.6939) & 0.953) & 0.8267) & 0.952) & 1.1974) & 0.951) & 1.8826) & 0.963) \\
		\multirow{2}[0]{*}{n=30, m=20, S=0.1} & (0.8191, & (0.952, & (0.8210, & (0.953, & (0.9861, & (0.951, & (1.1686, & (0.959, \\
		& 0.8541) & 0.953) & 0.9678) & 0.951) & 1.7118) & 0.951) & 2.4091) & 0.968) \\
		\multirow{2}[0]{*}{n=30, m=15, S=0.1} & (1.0720, & (0.951, & (0.9854, & (0.948, & (1.0137, & (0.951, & (1.3746, & (0.954, \\
		& 1.5067) & 0.951) & 1.3004) & 0.965) & 2.3050) & 0.951) & 3.4663) & 0.958) \\
		\hline
		\multirow{2}[0]{*}{n=15, m=15, S=0.1} & (0.8566, & (0.954, & (0.8816, & (0.945, & (1.0880, & (0.953, & (1.3225, & (0.945, \\
		& 0.9960) & 0.951) & 1.1293) & 0.934) & 1.8097) & 0.951) & 2.5000) & 0.965) \\
		\multirow{2}[0]{*}{n=15, m=12, S=0.1} & (0.9598, & (0.951, & (1.0852, & (0.959, & (1.3420, & (0.951, & (1.5635, & (0.95, \\
		& 1.0907) & 0.952) & 1.2963) & 0.94)  & 1.7570) & 0.951) & 3.2631) & 0.962) \\
		\multirow{2}[0]{*}{n=15, m=10, S=0.1} & (1.1078, & (0.951, & (1.2523, & (0.961, & (1.4652, & (0.951, & (1.7649, & (0.952, \\
		& 1.1124) & 0.951) & 1.5622) & 0.947) & 2.1629) & 0.951) & 4.2863) & 0.959) \\
		\multirow{2}[0]{*}{n=15, m=7, S=0.1} & (1.3373, & (0.951, & (1.6139, & (0.959, & (1.4974, & (0.951, & (2.1922, & (0.958, \\
		& 1.8838) & 0.951) & 2.7710) & 0.963) & 2.8849) & 0.951) & 7.8173) & 0.969) \\
		\hline
	\end{tabular}%
}
	\label{ci1}%
\end{table}%

\begin{table}[htbp]
	\centering
	\caption{Confidence length and coverage probability for the parameters ($\gamma$,$\delta$) under different censoring schemes.}
	\resizebox{15cm}{!}{
	\renewcommand*{\arraystretch}{1.3}
	\begin{tabular}{ccc|cc|cc|cc}
		\hline
		& \multicolumn{4}{c|}{$\gamma=1$, $\delta=1$} & \multicolumn{4}{c}{$\gamma=1.5$, $\delta=2$} \\ \hline
		& \multicolumn{2}{c|}{Bayesian HPDC} & \multicolumn{2}{c|}{MLE} & \multicolumn{2}{c|}{Bayesian HPDC} & \multicolumn{2}{c}{MLE} \\
		\hline
		         & \multicolumn{1}{c}{CL} & \multicolumn{1}{c|}{CP} & \multicolumn{1}{c}{CL} & \multicolumn{1}{c|}{CP} & \multicolumn{1}{c}{CL} & \multicolumn{1}{c|}{CP} & \multicolumn{1}{c}{CL} & \multicolumn{1}{c}{CP} \\
		\cline{2-9}    \multicolumn{1}{c}{\multirow{2}[1]{*}{n=100, m=100, S=0.2}} & (0.3179, & (0.951,  & (0.3099,  & (0.952,  & (0.5288,  & (0.952,  & (0.4648,  & (0.952,  \\
		&  0.3562) & 0.952) & 0.4158) & 0.96) & 0.8519) & 0.952) & 0.8187) & 0.957) \\
		\multirow{2}[0]{*}{n=100, m=90, S=0.2} & (0.3660, & (0.951,  & (0.3411,  & (0.948,  & (0.5274,  & (0.956,  & (0.4953,  & (0.947,  \\
		&  0.4885) & 0.952) & 0.4229) & 0.954) & 0.9306) & 0.952) & 0.8611) & 0.952) \\
		\multirow{2}[0]{*}{n=100, m=85, S=0.2} & (0.3942,  & (0.952,  & (0.3553,  & (0.953,  & (0.5649,  & (0.956,  & (0.5108,  & (0.951,  \\
		& 0.4954) & 0.954) & 0.4288) & 0.953) & 0.9012) & 0.951) & 0.8890) & 0.951) \\
		\multirow{2}[0]{*}{n=100, m=80, S=0.2} & (0.4006,  & (0.951,  & (0.3695,  & (0.951,  & (0.6132,  & (0.951,  & (0.5261,  & (0.947,  \\
		& 0.4243) & 0.952) & 0.4368) & 0.953) & 0.9005) & 0.952) & 0.9200) & 0.95) \\
		\hline
		
		\multirow{2}[0]{*}{n=60, m=60, S=0.2} & (0.3241,  & (0.957,  & (0.4037,  & (0.954,  & (0.6451,  & (0.951,  & (0.6056,  & (0.954,  \\
		& 0.4537) & 0.952) & 0.5375) & 0.948) & 1.0303) & 0.953) & 1.0670) & 0.959) \\
		\multirow{2}[0]{*}{n=60, m=55, S=0.2} & (0.3853, & (0.952,  & (0.4389,  & (0.953,  & (0.6768,  & (0.954,  & (0.6394,  & (0.956,  \\
		&  0.5344) & 0.954) & 0.5454) & 0.954) & 1.1031) & 0.954) & 1.1143) & 0.961) \\
		\multirow{2}[0]{*}{n=60, m=50, S=0.2} & (0.5138,  & (0.951,  & (0.4709,  & (0.96,  & (0.7775,  & (0.951,  & (0.6738,  & (0.955,  \\
		& 0.6211) & 0.953) & 0.5598) & 0.95) & 1.1386) & 0.953) & 1.1795) & 0.96) \\
		\multirow{2}[0]{*}{n=60, m=45, S=0.2} & (0.4386,  & (0.952,  & (0.5034,  & (0.952,  & (0.7577,  & (0.951,  & (0.7089,  & (0.957,  \\
		& 0.6178) & 0.953) & 0.5818) & 0.956) & 1.2459) & 0.951) & 1.2593) & 0.962) \\
		\hline
		
		\multirow{2}[0]{*}{n=30, m=30, S=0.2} & (0.5174,  & (0.953,  & (0.5890,  & (0.947,  & (0.8078,  & (0.951,  & (0.8836,  & (0.947,  \\
		& 0.7870) & 0.951) & 0.7749) & 0.943) & 1.1308) & 0.951) & 1.5953) & 0.966) \\
		\multirow{2}[0]{*}{n=30, m=25, S=0.2} & (0.5693, & (0.951,  & (0.6851,  & (0.959,  & (0.7871,  & (0.953,  & (0.9777,  & (0.953,  \\
		&  0.7666) & 0.951) & 0.8172) & 0.951) & 1.2227) & 0.951) & 1.7784) & 0.963) \\
		\multirow{2}[0]{*}{n=30, m=20, S=0.2} & (0.8479, & (0.951,  & (0.7903,  & (0.954,  & (0.8870,  & (0.951,  & (1.0878,  & (0.961,  \\
		&  0.7257) & 0.951) & 0.9183) & 0.949) & 1.4659) & 0.953) & 2.0885) & 0.97) \\
		\multirow{2}[0]{*}{n=30, m=15, S=0.2} & (0.9756,  & (0.953,  & (0.9230,  & (0.959,  & (1.0329,  & (0.951,  & (1.2294,  & (0.953,  \\
		& 1.0953) & 0.951) & 1.1366) & 0.951) & 1.8774) & 0.951) & 2.6186) & 0.959) \\
		\hline
		
		\multirow{2}[0]{*}{n=15, m=15, S=0.2} & (0.8566,  & (0.954,  & (0.8816,  & (0.945,  & (1.0879,  & (0.953,  & (1.3224,  & (0.945,  \\
		& 0.9959) & 0.951) & 1.1293) & 0.934) & 1.8096) & 0.951) & 2.5000) & 0.965) \\
		\multirow{2}[0]{*}{n=15, m=12, S=0.2} & (0.9949,  & (0.951,  & (1.0575,  & (0.958,  & (1.1353,   & (0.951,  & (1.4842,  & (0.947,  \\
		& 0.9821) & 0.951) & 1.2549) & 0.939) & 1.8291) & 0.951) & 2.9176) & 0.964) \\
		\multirow{2}[0]{*}{n=15, m=10, S=0.2} & (1.0483, & (0.952,  & (1.1943,  & (0.955,  & (1.3200,  & (0.951,  & (1.6201,  & (0.959,  \\
		&  1.0680) & 0.952) & 1.4344) & 0.944) & 1.9196) & 0.952) & 3.4056) & 0.966) \\
		\multirow{2}[0]{*}{n=15, m=7, S=0.2} & (1.2669,  & (0.952,  & (1.4606,  & (0.964,  & (1.4345,  & (0.952,  & (1.8905,  & (0.955,  \\
		& 1.5338) & 0.951) & 1.9870) & 0.959) & 2.2919) & 0.951) & 4.6995) & 0.962)\\
		\hline
	\end{tabular}%
}
	\label{ci2}%
\end{table}%


\section{Real Data Example}

For the real-data application, we use the data from Hinkley \cite{hinkley1977quick}, which consists of 30 values of March precipitation readings (in inches) for Minneapolis/St. Paul. The data set is provided in Table \ref{dataset}. First, it must be verified that this data set can be used to analyze the Weibull distribution. Also, for the sake of comparison, Lindley and Inverse Weibull (IW) distributions are also taken into consideration. We looked at the Kolmogorov-Simirnov (KS) test to determine the quality of fit. In order to confirm the data's suitability, other criteria are also taken into account, including the Akaike information criterion (AIC), the corrected Akaike information criterion (AICC), the Bayesian information criterion (BIC), the Hannan-Quinn criterion (HQC), and negative log-likelihood (NLL). The distribution with the smallest estimated values of the KS statistic, AIC, AICC, BIC, HQC, and NLL and a comparatively high p-value is the better lifetime model. Table \ref{datafit} indicates that the Weibull model fits the data far better than the other models. Additionally, we provide fitted density, empirical distribution function, PP-plot and QQ-plot in Figure \ref{fittingplot}. Thus, Weibull distribution is used to draw conclusions about the data under consideration.

\begin{table}[htbp]
	\centering
	\caption{Data of 30 successive values of March precipitation}
	\label{dataset}
		\begin{tabular}{cccccccccc}
	\toprule
	0.77 & 1.74 & 0.81 & 1.20 & 1.95 & 1.20 & 0.47 & 1.43 & 3.37 & 2.20 \\
	3.00 & 3.09 & 1.51 & 2.10 & 0.52 & 1.62 & 1.31 & 0.32 & 0.59 & 0.81 \\
	2.81 & 1.87 & 1.18 & 1.35 & 4.75 & 2.48 & 0.96 & 1.89 & 0.90 & 2.05 \\
	\toprule
\end{tabular}
\end{table}

\begin{table}[htbp]
	\centering
	\caption{The goodness-of-fit statistics for the data set}
	\resizebox{15cm}{!}{
		\renewcommand*{\arraystretch}{1.5}
	\begin{tabular}{|c|c|c|c|c|c|c|c|c|c|}
		\hline
		\textbf{Distribution} & \multicolumn{2}{c|}{\textbf{Estimates}} & \textbf{KS} & \textbf{p-value} & \textbf{AIC}   & \textbf{AICC}  & \textbf{BIC}   & \textbf{HQC}   & \textbf{NLL} \\
		\hline
		Weibull & 1.8089 & 0.3155 & 0.0689 & 0.9988 & 81.2866 & 81.7310 & 84.0890 & 82.1831 & 77.2866 \\
		\hline
		Lindley & 0.9096 &       & 0.1882 & 0.2383 & 88.2875 & 88.4303 & 89.6886 & 88.7357 & 86.2875 \\
		\hline
		IW & 1.0162 & 1.5496 & 0.1524 & 0.4893 & 154.9602 & 155.4046 & 157.7626 & 155.8567 & 150.9602 \\
		\hline
	\end{tabular}%
}
\label{datafit}
\end{table}%

\begin{figure}[ht] 
	\caption{Probability plots from the fitted Weibull distribution}
	\label{fittingplot} 
	\begin{minipage}[b]{0.5\linewidth}
		\centering
		\includegraphics[width=1\linewidth]{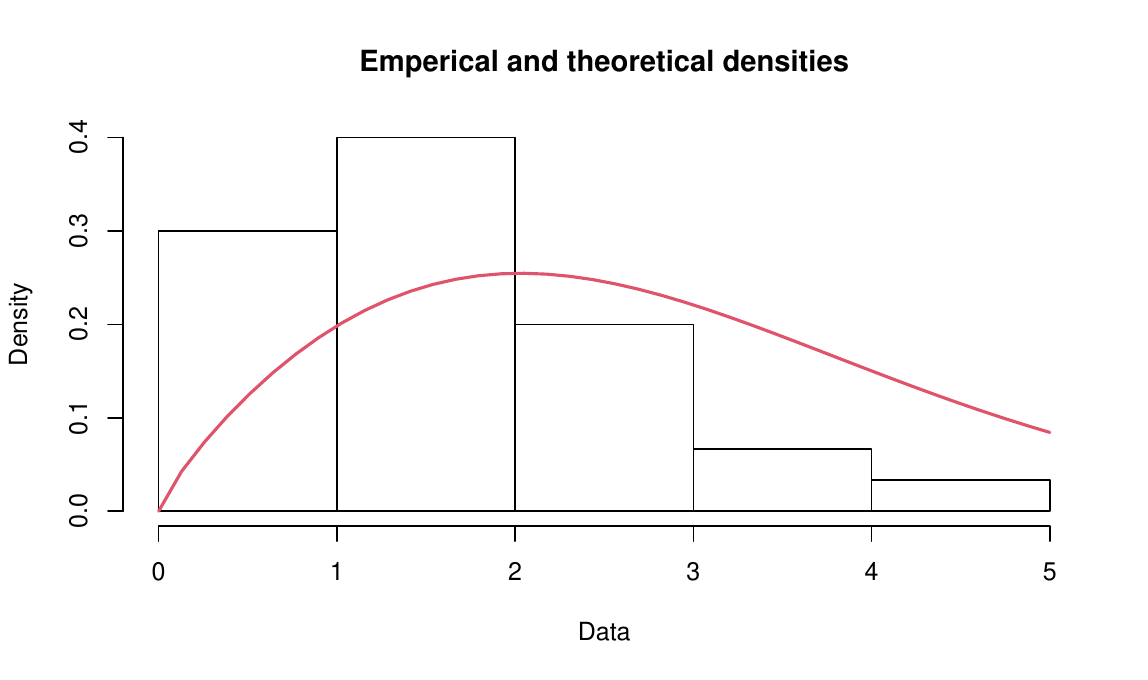} 
		\vspace{4ex}
	\end{minipage}
	\begin{minipage}[b]{0.5\linewidth}
		\centering
		\includegraphics[width=1\linewidth]{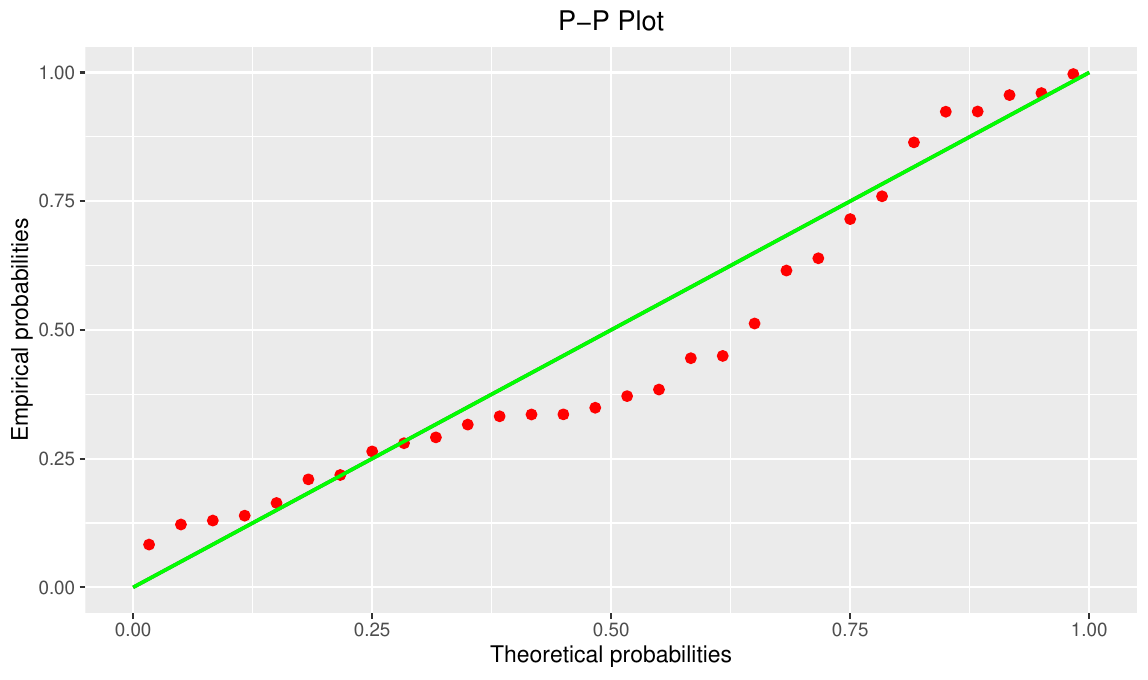} 
		\vspace{4ex}
	\end{minipage} 
	\begin{minipage}[b]{0.5\linewidth}
		\centering
		\includegraphics[width=1\linewidth]{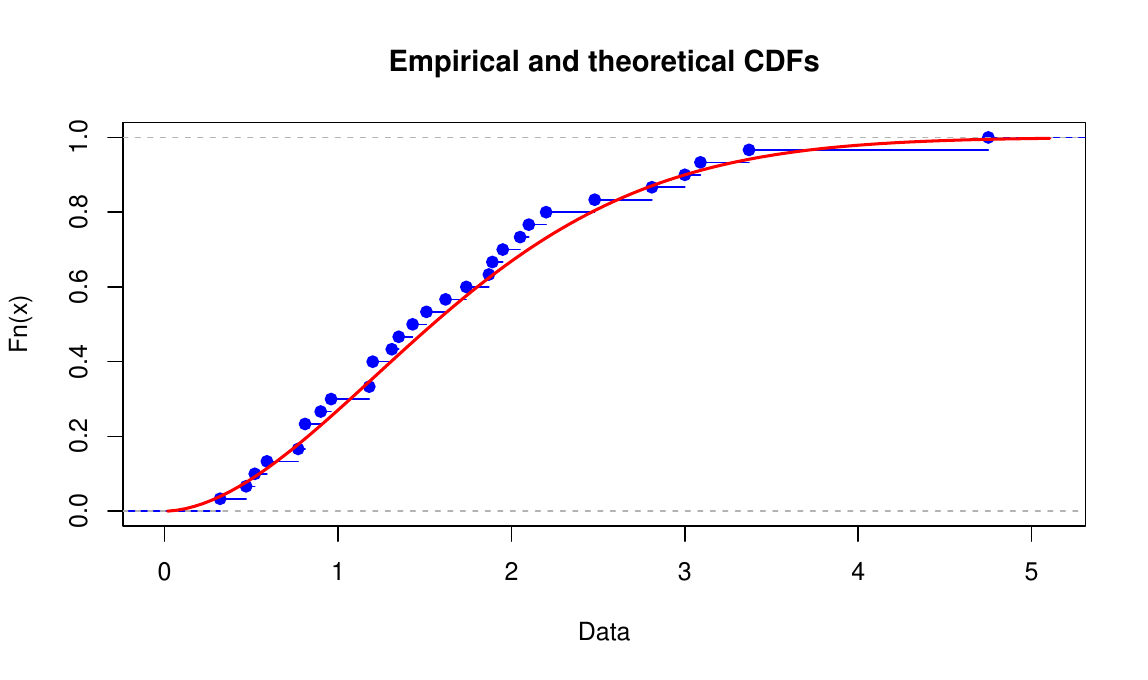} 
		\vspace{4ex}
	\end{minipage}
	\begin{minipage}[b]{0.5\linewidth}
		\centering
		\includegraphics[width=1\linewidth]{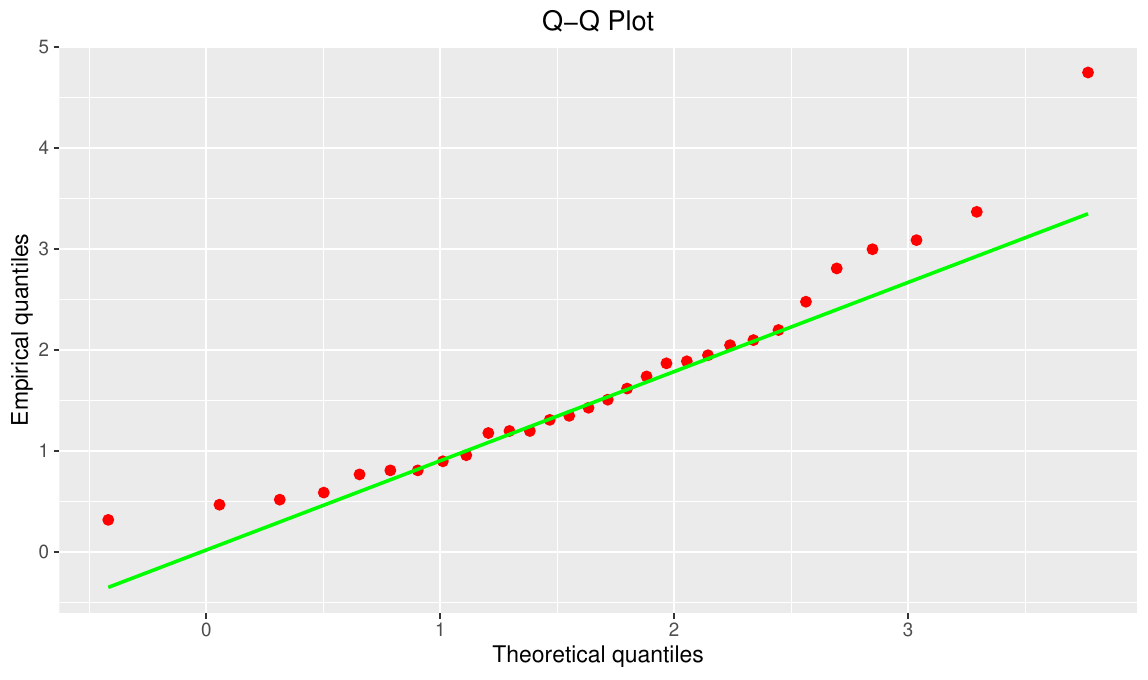} 
		\vspace{4ex}
	\end{minipage} 
\end{figure}

In order to analyse the data set under the T1T2 censoring scheme, we created artificial data sets from the real data for the schemes, $(m, S)=(20, 1), (20, 2), (15, 1)$ and $(15, 2)$. For each sample, MLE, the associated asymptotic confidence interval (ACI), the Bayes estimates under the squared error loss function, and the HPDC are calculated and are provided in the Table \ref{dataest}. Bayes estimates are calculated with the assumption that all hyperparameters are zeros. For each censoring scheme, MLE and Bayes estimates are nearly, but not quite, identical. In the case of confidence intervals, the HPDC is shorter than the ACI. For a fixed value of $m$, interval length decreases with supplymentary time S. The interval length of the scale parameter is shorter than that of the shape parameter.

\begin{table}[htbp]
	\centering
	\caption{Estimate values under different censoring schemes}
	\begin{tabular}{ccccc}
		\toprule
		\textbf{Scheme} & \multicolumn{2}{c}{\textbf{MLE}} & \multicolumn{2}{c}{\textbf{Bayesian}} \\
		\cmidrule{2-5}    \textbf{(m, S)} & \textbf{Estimates} & \textbf{ACI}   & \textbf{Estimates} & \textbf{HPDC} \\
		\midrule
		\multirow{2}[2]{*}{(20,1)} & $\hat{\gamma}$=1.8461 & (1.2618, 2.4304) & $\hat{\gamma}$=1.9131 & (1.4346, 2.4902) \\
		& $\hat{\delta}$=0.3099 & (0.1307, 0.4891) & $\hat{\delta}$=0.2989 & (0.1677, 0.5092) \\
		\midrule
		\multirow{2}[2]{*}{(20,2)} & $\hat{\gamma}$=1.8534 & (1.3243, 2.3826) & $\hat{\gamma}$=1.7649 & (1.3577, 2.2958) \\
		& $\hat{\delta}$=0.3105 & (0.1327, 0.4882) & $\hat{\delta}$=0.3491 & (0.1813, 0.5321) \\
		\midrule
		\multirow{2}[2]{*}{(15,1)} & $\hat{\gamma}$=1.9386 & (1.2786, 2.5985) & $\hat{\gamma}$=1.887 & ( 1.2267, 2.4044) \\
		& $\hat{\delta}$=0.3042 & (0.1259, 0.4825) & $\hat{\delta}$= 0.3269 & (0.1614, 0.4939) \\
		\midrule
		\multirow{2}[2]{*}{(15,2)} & $\hat{\gamma}$=1.9174 & (1.3644, 2.4704) & $\hat{\gamma}$=1.8477 & (1.3800, 2.3001) \\
		& $\hat{\delta}$=0.3031 & (0.1272, 0.4791) & $\hat{\delta}$=0.3325 & (0.2195, 0.5441) \\
		\bottomrule
	\end{tabular}%
	\label{dataest}%
\end{table}%

\section{Conclusion}
In this paper, we have introduced the T1-T2 mixture censoring scheme and examined its statistical inference based on the Weibull distribution. The new scheme guarantees a maximum possible number of failures within the prefixed time limit. The expected value of the failures and the expected duration of experiment are derived and it has been seen that both are invariant under scalar multiplication. The Fisher information for the new censoring scheme is presented.
The performances of MLE and Bayesian estimates for squared error and LINEX loss function,  demonstrate their effectiveness in providing reliable information. A real data is used to illustrate the functionality of the T1-T2 mixture censoring scheme.

\clearpage
\bibliographystyle{plain}
\bibliography{mixbib}

\begin{thebibliography}{10}
\providecommand{\url}[1]{\normalfont{#1}}
\providecommand{\urlprefix}{Available from: }

\bibitem{kumaraswamy1980generalized}
Kumaraswamy~P. A generalized probability density function for double-bounded
  random processes. Journal of hydrology. 1980;\hspace{0pt}46(1-2):79--88.

\bibitem{dey2017estimation}
Dey~S, Mazucheli~J, Anis~M. Estimation of reliability of multicomponent
  stress--strength for a kumaraswamy distribution. Communications in
  Statistics-Theory and Methods. 2017;\hspace{0pt}46(4):1560--1572.

\bibitem{kizilaslan2018estimation}
K{\i}z{\i}laslan~F, Nadar~M. Estimation of reliability in a multicomponent
  stress--strength model based on a bivariate kumaraswamy distribution.
  Statistical Papers. 2018;\hspace{0pt}59(1):307--340.

\bibitem{kohansal2021estimation}
Kohansal~A, Bakouch~HS. Estimation procedures for kumaraswamy distribution
  parameters under adaptive type-ii hybrid progressive censoring.
  Communications in Statistics-Simulation and Computation.
  2021;\hspace{0pt}50(12):4059--4078.

\bibitem{usman2020marshall}
Usman~RM, ul~Haq~MA. The marshall-olkin extended inverted kumaraswamy
  distribution: Theory and applications. Journal of King Saud
  University-Science. 2020;\hspace{0pt}32(1):356--365.

\bibitem{al2021estimation}
Al-Babtain~AA, Elbatal~I, Chesneau~C, et~al. Estimation of different types of
  entropies for the kumaraswamy distribution. Plos one.
  2021;\hspace{0pt}16(3):e0249027.

\bibitem{sherwania2021transmuted}
Sherwania~RAK, Waqas~M, Saeed~N, et~al. Transmuted inverted kumaraswamy
  distribution: Theory and applications. Punjab University Journal of
  Mathematics. 2021;\hspace{0pt}53(3).

\bibitem{nagarjuna2021kumaraswamy}
Nagarjuna~VB, Vardhan~RV, Chesneau~C. Kumaraswamy generalized power lomax
  distributionand its applications. Stats. 2021;\hspace{0pt}4(1):28--45.

\bibitem{abd2017inverted}
Abd AL-Fattah~A, El-Helbawy~A, Al-Dayian~G. Inverted kumaraswamy distribution:
  Properties and estimation. Pakistan Journal of Statistics.
  2017;\hspace{0pt}33(1).

\bibitem{shahbaz2012kumaraswamy}
Shahbaz~MQ, Shahbaz~S, Butt~NS. The kumaraswamy-inverse weibull distribution.
  Shahbaz, MQ, Shahbaz, S, \& Butt, NS (2012) The Kumaraswamy--Inverse Weibull
  Distribution Pakistan journal of statistics and operation research.
  2012;\hspace{0pt}8(3):479--489.

\bibitem{iqbal2017generalized}
Iqbal~Z, Tahir~MM, Riaz~N, et~al. Generalized inverted kumaraswamy
  distribution: properties and application. Open Journal of Statistics.
  2017;\hspace{0pt}7(04):645.

\bibitem{jamal2019generalized}
Jamal~F, Arslan~Nasir~M, Ozel~G, et~al. Generalized inverted kumaraswamy
  generated family of distributions: theory and applications. Journal of
  Applied Statistics. 2019;\hspace{0pt}46(16):2927--2944.

\bibitem{rashad2019}
Rashad~A, Yusuf~M, Moheb~S. Approximate bayes estimators of the inverted
  kumarswamy distribution parameters based on progressive type-ii censoring
  scheme. Journal of Statistics Applications \& Probability.
  2019;\hspace{0pt}8(3):189--199.

\bibitem{hameed2020estimation}
Hameed~BA, Salman~AN, Kalaf~BA. On estimation of in cased inverse kumaraswamy
  distribution. Iraqi Journal of Science. 2020;\hspace{0pt}61(4):845--853.

\bibitem{aly2020multivariate}
Aly~HM, Abuelamayem~OA. Multivariate inverted kumaraswamy distribution:
  Derivation and estimation. Mathematical Problems in Engineering.
  2020;\hspace{0pt}2020.

\bibitem{kumar2015method}
Kumar~D, Singh~U, Singh~SK. A method of proposing new distribution and its
  application to bladder cancer patients data. J Stat Appl Pro Lett.
  2015;\hspace{0pt}2(3):235--245.

\bibitem{maurya2017new}
Maurya~SK, Kaushik~A, Singh~SK, et~al. A new class of exponential transformed
  lindley distribution and its application to yarn data. International Journal
  of Statistics and Economics. 2017;\hspace{0pt}18(2):135--151.

\bibitem{tripathi2021inferences}
Tripathi~A, Singh~U, Singh~SK. Inferences for the dus-exponential distribution
  based on upper record values. Annals of Data Science.
  2021;\hspace{0pt}8(2):387--403.

\bibitem{2021}
Karakaya~K, Kinaci~I, KUŞ~C, et~al. On the dus-kumaraswamy distribution.
  2021;\hspace{0pt}13(29–38).

\bibitem{deepthi2020upside}
Deepthi~K, Chacko~V. An upside-down bathtub-shaped failure rate model using a
  dus transformation of lomax distribution. In: Stochastic models in
  reliability engineering. CRC Press; 2020. p. 81--100.

\bibitem{kavya2020generalized}
Kavya~P, Manoharan~M. On a generalized lifetime model using dus transformation.
  In: Applied probability and stochastic processes. Springer; 2020. p.
  281--291.

\bibitem{gauthami2021dus}
Gauthami~P, Chacko~V. Dus transformation of inverse weibull distribution: An
  upside-down failure rate model. Reliability: Theory \& Applications.
  2021;\hspace{0pt}16(2 (62)):58--71.

\bibitem{dumonceaux1973discrimination}
Dumonceaux~R, Antle~CE. Discrimination between the log-normal and the weibull
  distributions. Technometrics. 1973;\hspace{0pt}15(4):923--926.

\bibitem{lieblein1956statistical}
Lieblein~J, Zelen~M. Statistical investigation of the fatigue life of
  deep-groove ball bearings. Journal of research of the national bureau of
  standards. 1956;\hspace{0pt}57(5):273--316.

\end{thebibliography}


\begin{thebibliography}{10}

\bibitem{burkschat2008optimality}
M~Burkschat.
\newblock On optimality of extremal schemes in progressive type ii censoring.
\newblock {\em Journal of Statistical Planning and Inference},
  138(6):1647--1659, 2008.

\bibitem{childs2003exact}
A~Childs, B~Chandrasekar, N~Balakrishnan, and D~Kundu.
\newblock Exact likelihood inference based on type-i and type-ii hybrid
  censored samples from the exponential distribution.
\newblock {\em Annals of the Institute of Statistical Mathematics},
  55(2):319--330, 2003.

\bibitem{cohen1963progressively}
A~Clifford Cohen.
\newblock Progressively censored samples in life testing.
\newblock {\em Technometrics}, 5(3):327--339, 1963.

\bibitem{efron1990fisher}
Bradley Efron and Iain~M Johnstone.
\newblock Fisher's information in terms of the hazard rate.
\newblock {\em The Annals of Statistics}, pages 38--62, 1990.

\bibitem{epstein1954truncated}
Benjamin Epstein.
\newblock Truncated life tests in the exponential case.
\newblock {\em The Annals of Mathematical Statistics}, pages 555--564, 1954.

\bibitem{hastings1970monte}
W~Keith Hastings.
\newblock Monte carlo sampling methods using markov chains and their
  applications.
\newblock 1970.

\bibitem{hinkley1977quick}
David Hinkley.
\newblock On quick choice of power transformation.
\newblock {\em Journal of the Royal Statistical Society: Series C (Applied
  Statistics)}, 26(1):67--69, 1977.

\bibitem{kundu2007hybrid}
Debasis Kundu.
\newblock On hybrid censored weibull distribution.
\newblock {\em Journal of Statistical Planning and Inference},
  137(7):2127--2142, 2007.

\bibitem{kundu2008bayesian}
Debasis Kundu.
\newblock Bayesian inference and life testing plan for the weibull distribution
  in presence of progressive censoring.
\newblock {\em Technometrics}, 50(2):144--154, 2008.

\bibitem{metropolis1953equation}
Nicholas Metropolis, Arianna~W Rosenbluth, Marshall~N Rosenbluth, Augusta~H
  Teller, and Edward Teller.
\newblock Equation of state calculations by fast computing machines.
\newblock {\em The journal of chemical physics}, 21(6):1087--1092, 1953.

\bibitem{ng2009statistical}
Hon Keung~Tony Ng, Debasis Kundu, and Ping~Shing Chan.
\newblock Statistical analysis of exponential lifetimes under an adaptive
  type-ii progressive censoring scheme.
\newblock {\em Naval Research Logistics (NRL)}, 56(8):687--698, 2009.

\bibitem{park2009simple}
Sangun Park and N~Balakrishnan.
\newblock On simple calculation of the fisher information in hybrid censoring
  schemes.
\newblock {\em Statistics \& Probability Letters}, 79(10):1311--1319, 2009.

\end{thebibliography}

\end{document}